\RequirePackage{amsmath}

\pdfoutput=1

\documentclass{iopart}
\usepackage{iopams}
\usepackage{amsmath,amsfonts,amsthm,amssymb}
\usepackage{color,graphicx}

\newcommand{\C}{\mathbf{C}}
\newcommand{\G}{\mathbf{G}}
\newcommand{\HH}{\mathbf{H}}
\newcommand{\Q}{\mathbf{Q}}
\newcommand{\R}{\mathbf{R}}
\newcommand{\Z}{\mathbf{Z}}
\newcommand{\CC}{\mathcal{C}}
\newcommand{\DD}{\mathcal{D}}
\newcommand{\PP}{\mathcal{P}}

\newcommand{\GL}{\mathbf{GL}}
\newcommand{\SO}{\mathbf{SO}}
\newcommand{\SU}{\mathbf{SU}}
\newcommand{\eps}{\varepsilon}
\newcommand{\lam}{\lambda}

\newcommand{\Om}{\Omega}
\newcommand{\tu}{\tilde u}
\newcommand{\tv}{\tilde v}
\def\ei#1{e^{i#1}}
\def\eim#1{e^{-i#1}}
\def\tr{\mathop{\mathrm{tr}}}
\def\diag{\mathop{\mathrm{diag}}}
\def\spann{\mathop{\mathrm{span}}}
\newcommand{\bsk}{\bigskip}
\newcommand{\msk}{\medskip}

\newcommand{\tp}{\otimes}
\newcommand{\wt}{\widetilde}
\newcommand{\cg}[2]{\langle #1 | #2 \rangle}
\newcommand{\threej}[2]{
\ensuremath{
\begin{pmatrix} 
 #1 \\ 
 #2 
\end{pmatrix}}}

\newtheorem{prop}{Proposition}
\newtheorem{theo}{Theorem}
\newtheorem{lemma}{Lemma}
\renewenvironment{proof}{\hfill\newline\noindent\textsc{Proof.}}{\hfill $\square$ \\}

\begin{document}
\title[Molien series and invariants for $\SO(3)\wr\Z_2$]{Molien series and low-degree invariants for a natural action of $\SO(3)\wr\Z_2$}

\author[1]{D. R. J. Chillingworth$^{1,2}$, R. Lauterbach$^3$ and S. S. Turzi$^4$}

\address{$^1$ Mathematical Sciences, University of Southampton, Southampton SO17 1BJ, UK \\
$^2$ BCAM -- Basque Center for Applied Mathematics, Mazarredo 14, E48009 Bilbao, Basque Country, Spain.}
\address{$^3$ Department Mathematik, Universit\"{a}t Hamburg, Bundesstra{\ss}e 55, 20146 Hamburg, Germany.}
\address{$^4$ Dipartimento di Matematica, Politecnico di Milano, Piazza Leonardo da Vinci 32, 20133 Milano, Italy.}
%
\begin{abstract} \noindent
We investigate the invariants of the $25$-dimensional real representation of the group \mbox{$\SO(3)\wr\Z_2$} given by the left and right actions of $\SO(3)$ on $5\times 5$ matrices together with matrix transposition; the action on column vectors is the irreducible $5$-dimensional representation of $\SO(3)$.  The $25$-dimensional representation arises naturally in the study of nematic liquid crystals, where the second-rank orientational order parameters of a molecule are represented by a symmetric $3\times3$ traceless symmetric matrix, and where a rigid rotation in $\R^3$ induces a linear transformation of this space of matrices. The entropy contribution to a free energy density function in this context turns out to have \mbox{$\SO(3)\wr\Z_2$} symmetry. Although it is unrealistic to expect to describe the complete algebraic structure of the ring of invariants, we are able to calculate the Molien series giving the number of linearly independent invariants at each homogeneous degree, and to express this 
as a rational 
function indicating the degrees of invariant polynomials that constitute a basis of 19 primary invariants. The algebra of invariants up to degree 4 is investigated in detail.
\end{abstract}
%
%
\section{Introduction and motivation}
Associated to any pair of representations $D_1,D_2$ of finite or compact Lie groups $\G_1,\G_2$ on finite-dimensional complex linear spaces $V_1,V_2$ respectively there is a natural action $D_1\tp D_2$ of $\G_1\times\G_2$ on the tensor product $V_1\tp V_2$ given by
\begin{equation}
D_1\tp D_2\,(g_1,g_2): v_1\tp v_2\mapsto D_1(g_1)v_1\tp D_2(g_2)v_2.
\end{equation}
If $D_1,D_2$ are irreducible then $D_1\tp D_2$ is irreducible, and indeed every irreducible representation of $\G_1\times\G_2$ is of this form~\cite[II (4.14)]{BT}. The space $V_1\tp V_2$ can be identified with the space $L(V_2^*,V_1)$ of linear maps from $V_2^*$ to $V_1$ by
\begin{equation}
v_1\tp v_2:V_2^*\to V_1:\alpha\mapsto\alpha(v_2)v_1
\end{equation}
and the representation $D_1\tp D_2$ corresponds to the action on $L(V_2^*,V_1)$ by left and right multiplication:
\begin{equation}
D_1\tp D_2\,(g_1,g_2):A\mapsto D_1(g_1)\,A\,D_2(g_2)^*
\end{equation}
for all $A\in L(V_2^*,V_1)$, where~${}^*$ denotes dual (space or linear map).
\begin{center}
\includegraphics[width=0.4\textwidth]{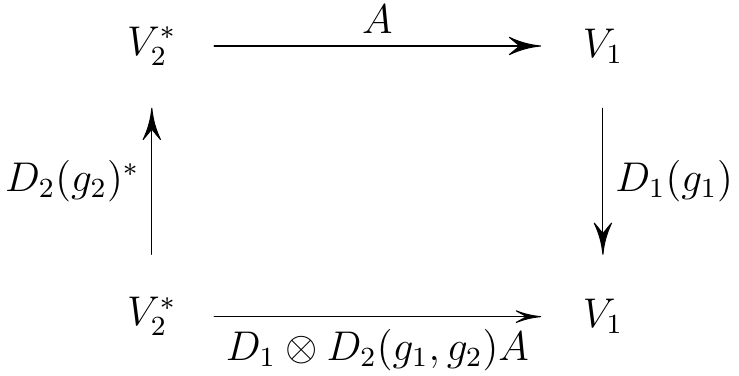}
\end{center}
If $V_2$ has an inner product and it is assumed (as may always be done) that $D_2$ acts by unitary transformations, then $V_2^*$ is identified with $V_2$ and $D_2(g_2)^*$ becomes the complex conjugate transposition of $D_2(g_2)$.
\subsection{Wreath product actions}  \label{ss:wreath}
In the particular case where $V_1=V_2=V$, $\G_1=\G_2=\G$ and $D_1=D_2=D$ then there is an extra symmetry inherent in the situation, namely the $\Z_2$-action of~${}^*$ on $L(V^*,V)$.  Combined with $D\tp D$ this gives an action 
$\wt D$ of the {\it  wreath product}
\begin{equation}
\mathbf{\Gamma}:=\G\wr\Z_2 \cong (\G\times\G) \rtimes \Z_2
\end{equation}
on $L(V^*,V)$.  Explicitly, each element of $\mathbf{\Gamma}$ can be written uniquely in the form $(g,h)\tau^\eps$ with $\eps\in\{0,1\}$, those elements with $\eps=0$ forming a subgroup $\mathbf{\Gamma}_0<\mathbf{\Gamma}$ of index~$2$ isomorphic to $\G\times \G$.  Elements of $\mathbf{\Gamma}_0$ act on $L(V^*,V)$ by 
\begin{equation}
\wt D(g,h): A\mapsto D(g)\,A\,D(h)^*
\end{equation}
while the elements of the other coset $\mathbf{\Gamma}_1$ act by 
\begin{equation}
\wt D(g,h)\tau: A\mapsto D(g)\,A^*\,D(h)^*.
\end{equation}
In this paper we focus on the case when $\G=\SO(3)$, and are interested in studying real-valued {\it invariants} for real wreath product actions as just described. 
\msk

It is well known that the irreducible representations of $\SO(3)$ are each characterised by a nonnegative integer $\ell$, there being (up to isomorphism) precisely one such representation $D^{(\ell)}$ for each given $\ell$ on a linear space $V^{(\ell)}$ of dimension $2\ell+1$.
Group actions of the wreath product type arise in certain physical applications such as the study of rotating ellipsoids~\cite{RS} (where $\ell=1$) or phase transitions of biaxial nematic liquid crystals~\cite{CD} (where $\ell=2$).  Invariants of these actions arise naturally from the physics when functions generated by microscopic considerations are averaged over the group orbits to give frame-independent macroscopic observables.
For example, in the molecular field theory for biaxial nematic liquid crystals as formulated in~\cite{LN} a fundamental role is played by a family of probability distributions $f$ on $\SO(3)$ of the form
\begin{equation}
f(\eta,\Om)= const.\times\exp \left(\eta\cdot Q(\Om)\right),
\label{eq:pdf}
\end{equation}
where $\Om\in\SO(3)$ and $Q(\Om)=D(\Om)\,Q$ with $D$ a real representation of  $\SO(3)$ on a space $V\cong\R^n$ of orientational order parameters $Q$ \cite{2011ST}.  Here $\eta\in L(V,V)$ plays the role of a set of Lagrange multipliers and the dot denotes the inner product $A\cdot B=\mathrm{trace}(A^*B)$ on $L(V,V)$.  The {\it entropy difference} function
$\Delta S:L(V,V)\to\R$ given by
\begin{equation}
\Delta S(\eta)=-k \int_{\SO(3)}f(\eta,\Om)\ln f(\eta,\Om) d\Om
\end{equation}
plays a central part in the study of phase transitions~\cite{LN} and is invariant under the action of $\SO(3)\wr\Z_2$ on $L(V,V)$ associated to $D$: this can be seen from the properties of the inner product and the fact that integrating over $R\in\SO(3)$ is the same as integrating over $R_0R$ (for fixed $R_0\in\SO(3)$) or over $R^T=R^{-1}$.  See~\cite{CD} for further details.
\msk 

To display the algebraic structure of the ring of invariant functions for group actions of this kind may be feasible (although complicated) for groups and representations of low dimension, as has been carried out for example in~\cite{KW} for a related case of $\G=\GL(2)$.
However, it is totally impractical in our case where $\SO(3)\wr \Z_2$ acts on the $25$-dimensional space $L(V,V)$ with $\dim V=5\,$: as we show below, there is an integrity basis for the ring of invariants having~19 primary invariants (to be expected, since $\dim\mathbf{\Gamma}=6$)  and 726,963,024 secondary invariants.  Nevertheless we are able to derive the Molien series that gives the number of linearly independent invariants at each homogeneous degree, both for the subgroup $\SO(3)\times\SO(3)$ and for the full group $\SO(3)\wr \Z_2$, and to display a rational Molien function that indicates at least some of the further structure at low degree.  In particular we give explicit generators at degrees $2,3$ and $4$ and in doing so we uncover some unexpected 
algebraic properties of tensor contractions that yield these generators. 
\subsection{Quaternions, $\SO(4)$ and $\SO(3)\times\SO(3)$}
The group $\SU(2)$ is isomorphic to the group $\Q^*$ of unit quaternions, and acts orthogonally on the space $\Q\cong\R^4$ of all quaternions by
\begin{equation}
q: x\mapsto qh(x)\bar q
\end{equation}
where we identify $x=(x_0,x_1,x_2,x_3)^T\in\R^4$ with
\begin{equation}
h(x)=x_0+x_1 i+x_2 j+ x_3k\in\Q.
\end{equation}
This action leaves the real subspace invariant and so induces an action on its orthogonal complement (pure quaternions) $\cong\R^3$, which defines a projection
\begin{equation}
\pi:\SU(2)\to \SO(3).
\end{equation}
It is straightforward to check that the kernel of $\pi$ is $\{E,-E\}$ where $E$ denotes the identity matrix in $\SU(2)$, so that the group $\SU(2)$ is a double cover of $\SO(3)$. Thus every representation of $\SO(3)$ induces one of $\SU(2)$, while the converse is true only if $-E$ acts trivially. 
\msk

The irreducible complex representations of $\SU(2)$ can be explicitly constructed on complex homogeneous polynomials in two variables of degree $d\ge 0$ (see e.g.~\cite[II\thinspace(5.3)]{BT}). Since the space of such polynomials has complex dimension $d+1$, we have such a representation for each positive integer dimension, and clearly the element $-E$ acts trivially if and only if $d$ is even.
Thus it is the odd-dimensional irreducible (complex) representations of $\SU(2)$ that give rise to the irreducible complex representations of $\SO(3)$ which, since these are of real type~\cite[II\thinspace6]{BT}, correspond to representations in each odd (real) dimension $2\ell+1$.  With this convention, the even-dimensional irreducible complex representations of $\SU(2)$ can be labelled by non-integer values of~$\ell$ which we call {\it half-integers}. 
\msk

Let 
\begin{equation}
{\rm Spin}_4=\Q^*\times \Q^*.
\end{equation}
There is a natural action of ${\rm Spin}_4$  on $\Q$ given by
\begin{equation}
(q_1,q_2):q\mapsto q_1q\bar q_2.
\end{equation}
Viewed as an action on $\R^4$, this is orthogonal and has determinant~$1$,
and thus defines a homomorphism
\begin{equation}
\tilde\pi:{\rm Spin}_4\to\SO(4).
\end{equation}
The kernel is again $\Z_2$, here generated by $(-E,-E)$ and so we have an isomorphism
\begin{equation}
{\rm Spin}_4/\Z_2=\Q^*\times \Q^*/\Z_2 \to \SO(4).
\end{equation}
At the same time we have an isomorphism
\begin{equation}
\SU(2)\times\SU(2)\,/\,\Z_2\times \Z_2 \to \SO(3)\times\SO(3),
\end{equation}
and so it is natural to ask: 
when does an $(\ell_1,\ell_2)$ representation of $\SU(2)\times\SU(2)$ correspond to a representation of $\SO(4)$? This is the case if and only if both numbers $\ell_1$, $\ell_2$ are half-integers or both are integers, since $(-E,-E)$ must act trivially and therefore $-E$ must act in the same way (change the sign or not change the sign) in both factors.  To summarise:
\begin{prop}
The tensor product of two (complex) irreducible representations of $\SU(2)$ with labels $\ell_1,\ell_2$ projects to a representation of $\SO(4)$
(of dimension $(2\ell_1+1)(2\ell_2+1)$) if and only if $\ell_1+\ell_2$ is an integer.  In particular, the tensor product of two (complex, real) irreducible representations of $\SO(3)$ always projects to a representation of $\SO(4)$.
\end{prop}

In the following table we list the labels for irreducible representations of $\SU(2)\times\SU(2)$ for dimensions $(2\ell_1+1)(2\ell_2+1)$ up to 33.
 \begin{equation}
 \begin{array}{cc|c|cc}
   \dim &(\ell_1,\ell_2)&&\dim&(\ell_1,\ell_2)\\
\hline
1&(0,0)&&19&(0,9),(9,0)\\
3&(0,1), (1,0)&&20&(\frac12,\frac92)\\
4&(\frac12,\frac12)&&21&(0,5),(1,3),(3,1),(5,0)\\
5&(0,2),(2,0)&&23&(0,11),(11,0)\\
7&(0,3),(3,0)&&24&(\frac12,\frac{11}{2}),(\frac32,\frac52)\\
8&(\frac12,\frac32)&&25&(0,12),(2,2),12,0)\\
9&(0,4),(1,1),(4,0)&&27&(1,4),(4,1)\\
11&(0,5),(5,0)&&28&(\frac12,\frac{13}{2})\\
12&(\frac12,\frac52)&&29&(0,14),(14,0)\\
13&(0,6),(6,0)&&30&(\frac12,\frac{15}2)\\
15&(0,7),(1,2),(7,0)&&31&(0,15),(15,0)\\
16&(\frac12,\frac72),(\frac32,\frac32)&&32&(\frac12,\frac{15}2),(\frac32,\frac72)\\
17&(0,8),(8,0)&&33&(0,16)(16,0)\\
 \end{array}
\end{equation}
Our interest in this paper is in the cases where $(\ell_1,\ell_2)$ are positive integers, and $\ell_1=\ell_2=\ell$ so that conjugate transposition
also acts on the tensor product representation.  For practical purposes this restricts us to $\ell=1$ or $\ell=2$. For further detailed information on subgroups of $\SO(4)$, quaternions and rotations see \cite{DV,Conway}.
\section{Molien functions}
Associated to any $n$-dimensional (real or complex) representation $D$ of a compact Lie group $\HH$ is its
{\it Molien integral}
\begin{equation}  \label{e:moliendef}
\PP_\HH=\int_\HH\det(I-tD(h))^{-1}dh   
\end{equation}
with respect to the normalised (Haar) invariant measure on $\HH$.  The integral has the remarkable property that when formally expanded as a power series in $t$ the coefficient of $t^m$ is a non-negative integer equal to the dimension of the linear space of $D$-invariant polynomials of homogeneous degree~$m$.  Moreover, when appropriately expressed as a rational function of~$t$, the Molien series conveys further information on the structure of the ring of invariants: see e.g.~\cite{MZ,ST}.  Since $\det(I-tD(g))=t^n\det(t^{-1}I-D(g))$ the integrand in~(\ref{e:moliendef}) can always be expressed in terms of the eigenvalues of $D(g)$ when these are known. In the case of a unitary representation the Molien series is automatically real.
\msk

We apply this to the case of $\HH={\mathbf\Gamma}=\G\wr\Z_2$ acting on $L(V,V)$ via $\wt D$ as in Section~\ref{ss:wreath}.  
First, we extend the action of $\HH$ to its associated action on the complex linear space $V\tp\C$ in order to be able to use a basis of (complex) eigenvectors for each element of $D(\HH)$. For $g\in\G$ let $\Lambda(g)\subset\C$ denote the corresponding set of eigenvalues $\{\lambda_1,\ldots,\lambda_n\}$ of $D(g)$ (with repeats) where $\dim V=n$.
\begin{prop} \label{p:eigen} 
Let $(M,N)\in {\mathbf\Gamma}_0\cong \G\times \G$.
\begin{enumerate}
\item The eigenvalues $\nu$ of $\wt D(g,h)$ are given by $\nu=\lam_i\bar\mu_j$ 
where $\lam_i,\mu_j\in\Lambda(g),\Lambda(h)$ respectively.
\item The eigenvalues $\gamma$ of $\wt D((g,h)\tau)$ are given by $\gamma^2=\bar\beta_i\bar\beta_j$ and by $\gamma=\bar\beta_i$
where $\beta_i,\beta_j\in\Lambda(gh)$.
\end{enumerate}
\end{prop}
\proof\quad  For simplicity we take $D$ as understood and drop it from the notation.
\begin{enumerate}
\item If $\{u_i\}$,$\{v_j\}$ are (complex) eigenbases for $g,h$ respectively ($i,j=1,\dots,n$) then the $n^2$ matrices $\{u_iv_j^*\}$ clearly form an eigenbasis for $(g,h)$ with eigenvalues $\nu$ as stated.
\item First note that the eigenvalues of $gh$ coincide with those of $hg$
since $w$ is an eigenvector of $gh$ if and only if $hw$ is an eigenvector of $hg$ with the same eigenvalue. Next, let $u,v$ be (complex) eigenvectors of $hg$ with distinct eigenvalues $\lam,\mu$ respectively. Write $\tu=gu$ and $\tv=gv$, so that $h\tu=\lam u$ and $h\tv=\mu v$.
We then see that
\begin{equation}
(g,h)\tau\,(\tu v^*)=gv\tu^*h^*=\bar\lam\tv u^*
\end{equation}
and
\begin{equation}
(g,h)\tau\,(\tv u^*)=gu\tv^*h^*=\bar\mu\tu v^*.
\end{equation}
Therefore the $2$-dimensional subspace $\spann\{\tu v^*,\tv u^*\}$ of $L(V\tp\C)$ is invariant under $(g,h)\tau$, and with respect to the basis $\{\tu v^*,\tv u^*\}$ the matrix of $(g,h)\tau$ is
\begin{equation}
\begin{pmatrix}
0 & \bar{\mu} \\ \bar{\lam} & 0
\end{pmatrix}
\end{equation}
whose eigenvalues $\gamma$ satisfy $\gamma^2=\bar\lam\bar\mu$ as claimed. We observe also that $\tu u^*$ is an eigenvector of $(g,h)\tau$ with eigenvalue $\bar\lam$. \endproof
\end{enumerate}
\msk

We apply this now to the case of $\G=\SO(3)$ acting on $V=V^{(\ell)}$ by the $D^{(\ell)}$ representation where $\ell=1,2$. The case $\ell=1$ is straightforward and doubtless well known in the invariant theory literature, but we include it here as illustration of the methods employed in the considerably more complicated $\ell=2$ case. Using Proposition~\ref{p:eigen} it is straightforward to write down the Molien integrand in~(\ref{e:moliendef}). The eigenvalues of $M\in \SO(3)$ are $\{1,\ei\theta,\eim\theta\}$
for some $\theta$ which we call the {\it rotation angle}.  Let $N\in \SO(3)$ with rotation angle
$\phi$, and suppose $MN$ has rotation angle $\psi$. Write (with new notation) $z=\ei\theta,w=\ei\phi$ and $v=\ei\psi$.
\subsection{The case $\ell=1$}
The eigenvalues of $D^{(1)}(M)$,$D^{(1)}(N)$ are those of $M,N$ and so by Proposition~\ref{p:eigen} the eigenvalues of $\wt D^{(1)}(M,N)$ are 
\begin{equation}
\{z^jw^k\}=\{\ei{(j\theta+k\phi)}\},\quad j,k\,\in\{-1,0,1\} 
\end{equation}
while the eigenvalues of $\wt D^{(1)}((M,N)\tau)$ are 
\begin{equation}
\{1,1,-1,\ei\psi,\eim\psi,\pm\ei{\psi/2},\pm\eim{\psi/2}\}.
\end{equation}
Therefore for $g=(M,N)\in {\mathbf\Gamma}_0$ we have
\begin{equation}
\det(I-t\wt D^{(1)}(g))=\prod_{j=-1}^1 \prod_{k=-1}^1 (1-tz^jw^k)
\end{equation}
while for $g=(M,N)\tau\in {\mathbf\Gamma}_1$
\begin{equation} \label{e:D1tau}
\hspace{-2cm}
\det(I-t\wt D^{(1)}(g))=(1-t)^2(1+t)(1-tv)(1-t\bar v)(1-t^2v)(1-t^2\bar v).
\end{equation}
\subsection{The case $\ell=2$}
The eigenvalues of $D^{(2)}(M)$ are $\{1,\ei\theta,\eim\theta,e^{2i\theta},e^{-2i\theta}\}$ 
and so the eigenvalues of $\wt D^{(2)}(M,N)$ are 
\begin{equation}
\{z^jw^k\}=\{\ei{(j\theta+k\phi)}\},\quad j,k\,\in\{-2,-1,0,1,2\} \notag
\end{equation}
while the eigenvalues of $\wt D^{(2)}((M,N)\tau)$ are 
\begin{align}
     1  &\qquad three\ times  \notag \\
    -1  &\qquad twice \notag \\
\pm e^{\pm i\psi/2} &\qquad  twice\ each  \notag\\
e^{\pm i\psi} &\qquad twice\ each  \notag\\
-e^{\pm i\psi} &\qquad once\ each  \notag\\
e^{\pm 2i\psi} &\qquad once\ each  \notag\\
\pm e^{\pm 3i\psi/2} &\qquad  once\ each \notag
\end{align}
(a total of $25$ eigenvalues). Therefore for $g\in {\mathbf\Gamma}_0$
\begin{equation}
\det(I-t\wt D^{(2)}(g))=\prod_{j=-2}^2 \prod_{k=-2}^2 (1-tz^jw^k)
\end{equation}
while for $g\in {\mathbf\Gamma}_1$ the expression becomes
\begin{align}
\det(I-t\wt D^{(2)}(g))&=(1-t)^3(1+t)^2(1-t^2v)^2(1-t^2\bar v)^2(1-tv)^2
      (1-t\bar v)^2\times  \notag \\ 
 &\qquad\times (1+tv)(1+t\bar v)(1-tv^2)(1-t\bar v^2)(1-t^2v^3)(1-t^2\bar v^3).\label{e:D2tau}
\end{align}
\subsection{Evaluation of the Molien integrals}
When the integration~(\ref{e:moliendef}) is carried out over the Lie group ${\mathbf\Gamma}=\G\wr\Z_2$ the result is half of the sum of integrals over the subgroup ${\mathbf\Gamma}_0$ and the coset ${\mathbf\Gamma}_1$:
\begin{equation}  \label{e:int+int}
\int_{\mathbf\Gamma}= \frac12\Bigl(\,\int_{{\mathbf\Gamma}_0} + 
\int_{{\mathbf\Gamma}_1}\Bigr)
\end{equation}
since the normalised Haar measure for ${\mathbf\Gamma}$ when restricted to ${\mathbf\Gamma}_0$ is half the Haar measure of ${\mathbf\Gamma}_0$. As described in~\cite[IV\thinspace(1.11)]{BT} for example, the Weyl Integral Formula shows that an integral over a compact Lie group can be decomposed into a double integral over a maximal torus and over the quotient of the group by this torus.  In the case when the integrand is a class function (invariant under conjugation in the group) then the latter integral becomes trivial, and the total integral reduces to an integral over the maximal torus at the cost of introducing a further term into the integrand that we call the {\it Weyl factor}, and dividing by the order $|W|$ of the {\it Weyl group} $W$. The integral over the maximal torus is expressed in terms of angular variables and is evaluated using residue calculus. Since the aim is to find $\PP_{\mathbf\Gamma}(t)$ as an infinite series in $t$ we can regard $t$ as a small real variable and in any case with $|t|<1$,
 and this 
restriction can be usefully exploited when evaluating residues. 
\msk

For the case of $\G=\SO(3)$ with which we are concerned, a maximal torus (circle) in $\G$ is given by the matrices
\begin{equation}
\begin{pmatrix}
 1 & 0 & 0 \\ 0 & \cos\theta & -\sin\theta \\  0 & \sin\theta & \cos\theta
\end{pmatrix}
\end{equation}
for $0\le\theta<2\pi$, and integration is over the unit circle regarded as the circle $|z|=1$ in the complex plane. The Weyl factor can be found to be 
\begin{equation}
\tfrac12(1-z)(1-\bar z)= (1-\cos\theta)
\end{equation}
(cf.~\cite{CL,LS}) while $W=\Z_2$ so $|W|=2$ and therefore the Molien integral~\eqref{e:moliendef} for the $\wt D^{(\ell)}$ representation of ${\mathbf\Gamma}_0$ becomes
\begin{equation}  \label{e:integrand0}
\PP_{{\mathbf\Gamma}_0}^\ell(t)=-\frac1{16\pi^2} \int_{|z|=1}\int_{|w|=1}
    \frac{(1-z)^2(1-w)^2}{\prod_{j,k=-\ell}^\ell(1-tz^jw^k)}\frac{dz}{z^2}\frac{dw}{w^2}
\end{equation}
for $\ell=1$ or $\ell=2$, with $4\pi^2$ the normalising factor of the invariant measure on the maximal torus $S^1\times S^1$ of ${\mathbf\Gamma}_0$.
\msk

For the second coset ${\mathbf\Gamma_1}$ of ${\mathbf\Gamma}$ in the case $\ell=1$ we have from~\eqref{e:D1tau} (with slight abuse of notation as ${\mathbf\Gamma}_1$ is not a group)
\begin{align}  
\PP_{{\mathbf\Gamma}_1}^1(t)& =\frac1{4\pi i}(1-t)^{-2}(1+t)^{-1} \times \notag \\
&\times \int_{|v|=1}\frac{(1-v)(1-v^{-1})}{(1-tv)(1-tv^{-1})(1-t^2v)(1-t^2v^{-1})}\frac{dv}v
\label{e:G1case1}
\end{align}
whereas for $\ell=2$ we have the slightly more complicated expression
\begin{equation} \label{e:G1case2}
\PP_{{\mathbf\Gamma}_1}^2(t)=\frac1{4\pi i}(1-t)^{-3}(1+t)^{-2}\int_{|v|=1}
\frac{(1-v)(1-v^{-1})}{K(v,t)}\frac{dv}v
\end{equation}
where from~\eqref{e:D2tau}
\begin{align}  \label{e:stardenom}
K(v,t)&=(1-tv)^2(1-t\bar v)^2(1+tv)(1+t\bar v)(1-tv^2)(1-t\bar v^2)\times \notag \\
    &\qquad\times (1-t^2v)^2(1-t^2\bar v)^2(1-t^2v^3)(1-t^2\bar v^3)\notag\\
    &=v^{-10}(1-tv)(1-tv^2)(1-t^2v)^2 (1-t^2v^2)(1-t^2v^3) \times \notag \\
    &\qquad\times (v-t)(v^2-t)(v^2-t^2)(v-t^2)^2(v^3-t^2).
\end{align}
Note that for $\mathbf\Gamma_1$ the calculation of the Molien integral \eqref{e:moliendef} reduces to a single integral since by~\eqref{e:D2tau} the expression \mbox{$\det\big[I-tD\big((M,N)\tau\big)\big]$} depends only the rotation angle $\psi$ of the product $MN$. More precisely, if $f(M,N)=\tilde{f}(L)$ with $L=MN$ then
\begin{equation} 
\int_{G\times G} f(M,N) d(M,N) 
= \int_{G} dN \int_{G}\tilde{f}(MN) dM = \int_{G}\tilde{f}(L) dL, 
\end{equation}
where we have used the facts that the measure is normalised and $G$-invariant. 
\subsection{Integration in the case $\ell=1$}
Since $|w|=1$ and $|t|<1$ the only terms in the denominator that give rise to nonzero residues
in the $z$-integral in $\PP_{{\mathbf\Gamma}_0}^1(t)$ are those with $j=-1$, that is $z=tw^k$ for $k=-1,0,1$, as well as $z=0$.  Thus 
\begin{equation} \label{e:wint1}
\PP_{{\mathbf\Gamma}_0}^1(t)=\frac1{8\pi i}\int_{|w|=1} \bigl(R_{-1}(w)+R_0(w)+R_1(w)\bigr)dw ,
\end{equation}
where $R_k(w)$ denotes the residue at $z=tw^k$ of the function $K(w)F_1(z,w)$ where
\begin{equation}
K(w)=\frac{(1-w)^2w^{-2}}{(1-tw^{-1})(1-t)(1-tw)}
\end{equation}
and
\begin{equation}
\hspace{-2.5cm}
F_1(z,w):=\frac{(1-z)^2z^{-2}}{(1-tz^{-1}w^{-1})(1-tz^{-1})(1-tz^{-1}w)(1-tzw^{-1})(1-tz)(1-tzw)}
\end{equation}
as it easy to verify that the residue at $z=0$ is zero.  It is straightforward (although a little tedious and best automated) to calculate that $R_k(w)=wS_k(w)^{-1}$ ($k=-1,0,1$)
where
\begin{align}
S_{-1}(w)&=t(t+1)(1-t)^2(w+1)(tw-1)(t^2-w)(w+t)      \\
S_0(w)&=t(t+1)(w-t^2)(w-t)(1-tw)(t^2w-1)             \\
S_1(w)&=t(t+1)(1-t)^2(w+1)(tw+1)(1-t^2w)(w-t).               
\end{align}
We thus see that the residues of~(\ref{e:wint1}) inside the unit circle occur where $w=\pm t,t^2$ and further calculation gives that
\begin{equation} \label{e:D1G0fn}
\hspace{-2cm}
\PP_{{\mathbf\Gamma}_0}^1(t)=\bigl((1-t)^3(t^2+1)(t^2+t+1)(t+1)^2\bigr)^{-1}
            =\bigl((1-t^2)(1-t^3)(1-t^4)\bigr)^{-1}
\end{equation}
which when expanded as a power series in $t$ gives
\begin{equation} \label{e:D1G0series}
\hspace{-2cm}
\PP_{{\mathbf\Gamma}_0}^1(t)=1+t^2+t^3+2t^4+t^5+3t^6+2t^7+4t^8+3t^9+5t^{10}+4t^{11}+7t^{12}+O(t^{13}).
\end{equation} 
However, the series~\eqref{e:D1G0series} is less informative than the rational function expression~\eqref{e:D1G0fn} which shows that there are three primary invariants of degrees $2,3,4$ and no secondary invariants.
Candidates for primary invariants are
\begin{equation}  \label{e:3invs}
\tr(A^TA),\quad \det(A),\quad \tr(A^TAA^TA).
\end{equation}
Since by the Cayley-Hamilton Theorem every $3\times3$ matrix $B$ satisfies its characteristic equation
\begin{equation}
B^3-\tr(B)\,B^2 + \tfrac12(\tr(B)^2-\tr(B^2))B+\det(B)\,I = 0,
\end{equation}
applying this to $B=A^TA$ shows how \lq natural' invariants of higher degree
such as $\tr(A^TAA^TAA^TA)$ may be expressed in terms of the primary generators~(\ref{e:3invs}).
\msk

The integral~\eqref{e:G1case1} for $\PP_{{\mathbf\Gamma}_1}^1(t)$ is simpler to evaluate.  The residues inside the unit circle are given by $v=t,t^2$ with a zero residue at $v=0$, and we find
\begin{equation}
\PP_{{\mathbf\Gamma}_1}^1(t)=\bigl((1-t^2)(1-t^3)(1-t^4)\bigr)^{-1}
\end{equation}
which, remarkably, is the same as $\PP_{{\mathbf\Gamma}_0}^1(t)$.  How can the Molien series
\begin{equation}
\PP_{{\mathbf\Gamma}}^1(t)
    =\frac12\bigl(\PP_{{\mathbf\Gamma}_0}^1(t)+\PP_{{\mathbf\Gamma}_1}^1(t)\bigr)  
\end{equation}
be the same as that of the proper subgroup $\mathbf\Gamma_0$?  This apparent anomaly is resolved by the observation in \cite{RS} that polar decomposition (every matrix can be written as the product of an orthogonal and a symmetric matrix) implies that $A^T$ lies on the same ${\mathbf\Gamma}_0$-orbit as $A$, and so every invariant for ${\mathbf\Gamma}_0$ is an invariant for $\tau$  and hence for ${\mathbf\Gamma}_0\cup{\mathbf\Gamma}_1$.  This simplification does not occur in the case $\ell=2$, however, because not every element of $\SO(5)$ is of the form $D^{(2)}(M)$ with $M\in \SO(3)$.
\subsection{Integration in the case $\ell=2$}
The calculations here become considerably more complicated, and hardly feasible to carry out by hand.  We have used MAPLE to assist with the residue calculations and the subsequent algebra.
\subsubsection{Residues of the $z$-singularities}
The terms in the denominator of the integrand~(\ref{e:integrand0}) for $\PP_{{\mathbf\Gamma}_0}^2(t)$ that contribute to residues for the $z$-integral are those with $j=-1,-2$. For $k=-2,\ldots,2$ we write $R_{1,k}(w)$ to denote the $z$-residue that occurs at $z=tw^k$ and $R_{2,k}(w)$ for the sum of the residues that occur at $z=\pm\sqrt{tw^k}$, and obtain the following results, most conveniently expressed in terms of the {\it inverses} of the $R_{1,k}$ and $R_{2,k}$:
\msk

\fbox{n=1}  
\msk

\begin{align}
-w^{29}t^4R_{1,-2}^{-1}&=T_{1,-2}(w,t)\,Q_{1,-2}(w,t) \\
-w^{17}t^4R_{1,-1}^{-1}&=T_{1,-1}(w,t)\,Q_{1,-1}(w,t) \\
-w^{13}t^4R_{1,0}^{-1}&=T_{1,0}(w,t)\,Q_{1,0}(w,t) \\
-w^{17}t^4R_{1,1}^{-1}&=T_{1,1}(w,t)\,Q_{1,1}(w,t) \\
-w^{29}t^4R_{1,2}^{-1}&=T_{1,2}(w,t)\,Q_{1,2}(w,t)
\end{align}
where
\begin{align}
\hspace{-0.5cm}T_{1,-2}(w,t)=T_{1,2}(w,t)&=(w^2+1)(w^2+w+1)(1-w)^2(w+1)^2(1+t)(1-t)^2 \\
\hspace{-0.5cm}T_{1,-1}(w,t)=T_{1,1}(w,t)&=(w+1)(w^2+w+1)(1-w)^2(1+t)(t^2+t+1)(1-t)^4 \\
T_{1,0}(w,t)&=(1-w)^2(w+1)^2(1+t)(t^2+t+1)(1-t)^2
\end{align}
and
\begin{align}
Q_{1,-2}&=(1-tw)(1-tw^2)(t-w)(t-w^3)(t-w^4)(t-w^5)(t-w^6) \times\notag \\
&\quad\times (t^2-w)(t^2-w^2)(t^2-w^3)(t^2-w^4)(t^3-w^2)(t^3-w^3)(t^3-w^4) \times\notag \\
&\qquad\times (t^3-w^5)(t^3-w^6)  \\
Q_{1,-1}&=(1-tw)(1-tw^2)(1-t^2w)(t-w^2)^2(t-w^3)(t-w^4)(t^2-w) \times\notag\\
      &\quad\times(t^2-w^2)(t^2-w^3)(t^3-w)(t^3-w^2)(t^3-w^3)(t^3-w^4) \\
Q_{1,0}&=(1-tw)^2(1-tw^2)^2(1-t^2w)(1-t^2w^2)(1-t^3w)(1-t^3w^2)(t-w)^2\times\notag\\
      &\quad\times(t-w^2)^2(t^2-w)(t^2-w^2)(t^3-w)(t^3-w^2) \\
Q_{1,1}&=(1-tw^2)^2(1-tw^3)(1-tw^4)(1-t^2w)(1-t^2w^2)(1-t^2w^3)\times\notag\\
      &\quad\times(1-t^3w)(1-t^3w^2)(1-t^3w^3)(1-t^3w^4)(t-w)(t-w^2)(t^2-w) \\
Q_{1,2}&=(1-tw)(1-tw^3)(1-tw^4)(1-tw^5)(1-tw^6)(1-t^2w)(1-t^2w^2) \times\notag\\
      &\quad\times (1-t^2w^3)(1-t^2w^4)(1-t^3w^2)(1-t^3w^3)(1-t^3w^4)(1-t^3w^5)\times \notag \\
& \qquad \times (1-t^3w^6)(t-w)(t-w^2).
\end{align}
\msk

\fbox{n=2}
\msk

\begin{align}
w^{17}P_{2,-2}(w,t)\,R_{2,-2}^{-1}&=T_{2,-2}(w,t)\,Q_{2,-2}(w,t) \\
w^{14}P_{2,-1}(w,t)\,R_{2,-1}^{-1}&=T_{2,-1}(w,t)\,Q_{2,-1}(w,t) \\
w^{13}P_{2,0}(w,t)\,R_{2,0}^{-1}&=T_{2,0}(w,t)\,Q_{2,0}(w,t) \\
w^{14}P_{2,1}(w,t)\,R_{2,1}^{-1}&=T_{2,1}(w,t)\,Q_{2,1}(w,t) \\
w^{17}P_{2,2}(w,t)\,R_{2,2}^{-1}&=T_{2,2}(w,t)\,Q_{2,2}(w,t)
\end{align}
where
\begin{align}
T_{2,-2}(w,t)=T_{2,2}(w,t)&=(w^2+1)(w^2+w+1)(w-1)^2(w+1)^2(1+t)\times \notag \\
&\quad\times(t^2+t+1)(-1+t)^4 \\
T_{2,-1}(w,t)=T_{2,1}(w,t)&=(w+1)(w^2+w+1)(w-1)^2(1+t)(1-t)^2 \\
T_{2,0}(w,t)&=(1-w)^2(w+1)^2(1+t)(t^2+t+1)(1-t)^2
\end{align}
and 
\begin{align}
Q_{2,-2}&=(1-tw)(1-tw^2)^2(1-tw^4)(1-tw^6)(1-t^3w^2) (t-w)(t^2-w)  \times\notag \\
      &\qquad\times(t^2-w^2)(t^2-w^3)(t^2-w^4)(t^3-w^2)(t^3-w^4)(t^3-w^6)  \\
Q_{2,-1}&=(1-tw)^2(1-tw^2)(1-tw^3)(1-tw^5)(1-t^2w)(1-t^3w)(1-t^3w^3) \times\notag \\
      &\qquad\times(t-w^2)(t-w^3)(t^2-w)(t^2-w^2)(t^2-w^3) \times \notag\\
      &\qquad \quad \times(t^3-w)(t^3-w^3)(t^3-w^5) \\
Q_{2,0}&=(1-tw)(1-tw^2)^2(1-tw^4)(1-t^2w)(1-t^2w^2)(1-t^3w^2)(1-t^3w^4)   \times\notag\\
      &\qquad\times(t-w)(t-w^2)^2(t-w^4)(t^2-w)(t^2-w^2)(t^3-w^2)(t^3-w^4) \\
Q_{2,1}&=(1-tw^2)(1-tw^3)(1-t^2w)(1-t^2w^2)(1-t^2w^3)\times\notag\\
      &\qquad\times(1-t^3w)(1-t^3w^3)(1-t^3w^5)(t-w)^2(t-w^2)(t-w^3)(t-w^5) \times\notag\\
      &\qquad\quad\times(t^2-w)(t^3-w)(t^3-w^3) \\
Q_{2,2}&=(1-tw)(1-t^2w)(1-t^2w^2)(1-t^2w^3)(1-t^2w^4)(1-t^3w^2)(1-t^3w^4)    \times\notag\\
      &\qquad\times (1-t^3w^6)(t-w)(t-w^2)^2(t-w^4)(t-w^6)(t^3-w^2)
\end{align}
while the $P_{2,k}(w,t)$ are irreducible polynomials in $w,t$ of total degree up to $18$ which we do not display here. 
\msk

For the purposes of this exposition, the important information lies in the $Q_{1,k}$ and $Q_{2,k}$ which indicate the points $w$ inside the unit circle at which residues will need to be evaluated. The relevant factors are those of the form $(t^q-w^p)$ for positive integers $p,q$.  We list in Table~\ref{tab:pq} the values of $(p,q)$ corresponding to each $Q_{j,k}$. 
\bsk

\begin{table}  [!ht]
   \begin{center}\scriptsize
     \begin{tabular}{ |c|l| }
\hline
\rule[-0.8ex]{0ex}{3.0ex} 
$Q_{1,-2}$& (1,1),\phantom{(2,1),}(3,1),(4,1),(5,1),(6,1),(1,2),(2,2),(3,2),(4,2),\phantom{(1,3),}(2,3),(3,3),(4,3),(5,3),(6,3) \\\hline
\rule[-0.8ex]{0ex}{3.0ex} $Q_{1,-1}$&\phantom{(1,1),}(2,1),(3,1),(4,1),\phantom{(5,1),(6,1),}(1,2),(2,2),(3,2),\phantom{(4,2),}(1,3),(2,3),(3,3),(4,3) \\\hline
\rule[-0.8ex]{0ex}{3.0ex}
$Q_{1,0}$&(1,1),(2,1),\phantom{(3,1),(4,1),(5,1),(6,1),}(1,2),(2,2),\phantom{(3,2),(4,2),}(1,3),(2,3) \\ \hline
\rule[-0.8ex]{0ex}{3.0ex}
$Q_{1,1}$&(1,1),(2,1),\phantom{(3,1),(4,1),(5,1),(6,1),}(1,2) \\ \hline
\rule[-0.8ex]{0ex}{3.0ex}
$Q_{1,2}$&(1,1),(2,1)  \\ \hline
\rule[-0.8ex]{0ex}{3.0ex}
$Q_{2,-2}$&(1,1),\phantom{(2,1),(3,1),(4,1),(5,1),(6,1),}(1,2),(2,2),(3,2),(4,2),\phantom{(1,3),}(2,3),\phantom{(3,3),}(4,3),\phantom{(5,3),}(6,3) \\\hline
\rule[-0.8ex]{0ex}{3.0ex}
$Q_{2,-1}$&\phantom{(1,1),}(2,1),(3,1),\phantom{(4,1),(5,1),(6,1),}(1,2),(2,2),(3,2),\phantom{(4,2),}(1,3),\phantom{(2,3),}(3,3),\phantom{(4,3),}(5,3) \\\hline
\rule[-0.8ex]{0ex}{3.0ex}
$Q_{2,0}$&(1,1),(2,1),\phantom{(3,1),}(4,1),\phantom{(5,1),(6,1),}(1,2),(2,2),\phantom{(3,2),(4,2),(1,3),}(2,3),\phantom{(3,3),}(4,3) \\\hline
\rule[-0.8ex]{0ex}{3.0ex}
$Q_{2,1}$&(1,1),(2,1),(3,1),\phantom{(4,1),}(5,1),\phantom{(6,1),}(1,2),\phantom{(2,2),(3,2),(4,2),}(1,3),\phantom{(2,3),}(3,3) \\\hline
\rule[-0.8ex]{0ex}{3.0ex}
$Q_{2,2}$&(1,1),(2,1),\phantom{(3,1),}(4,1),\phantom{(5,1),}(6,1),\phantom{(2,1),(2,2),(3,2),(4,2),(1,3),}(2,3) \\\hline
\end{tabular}
\end{center}    
 \caption{Table of values $p,q$ for which the factor $(t^q-w^p)$ occurs in $Q_{j,k}$ \\ for $j=1,2$ and $k=-2\ldots2$.}  \label{tab:pq}
\end{table}
\begin{equation}
F(w,t)=\sum_{k=-2}^2 R_{1,k} + \sum_{k=-2}^2 R_{2,k}
\end{equation}
at the singularities $w$ lying inside the unit circle $|w|=1$, where $|t|<1$.  It is helpful to note that some of the relevant factors in the denominators of the $R_{j,k}$ no longer appear when the $R_{j,k}$ are summed and simplified.  In fact it turns out that
\begin{equation} 
F(w,t)=2w^{13}(1-w)^2\frac{P(w,t)}{Q(w,t)}
\end{equation}
where $P(w,t)$ is an irreducible polynomial in $w,t$ of total degree $55$ and integer coefficients, while
\begin{align}
Q(w,t)&=(\cdots)(t-w)(t^2-w)(t^2-w^2)(t^2-w^3)(t^2-w^4)(t^3-w)(t^3-w^2)\times\notag \\
      &\qquad\qquad\times(t^3-w^3)(t^3-w^4)(t^3-w^5)(t^3-w^6)
\end{align}
with the dots $(\cdots)$ representing factors involving $t$ only or otherwise not vanishing for $w$ inside the unit circle.  Note that factors with $(p,q)=(m,1)$ for $3\le m \le 6$ do not arise here although they do appear in Table~\ref{tab:pq}: the residues of the individual $R_{j,k}$ corresponding to these factors sum to zero. The case $(p,q)=(2,1)$ also does not appear explicitly in $Q(w,t)$, although $(t-w^2)$ is a factor of both $(t^2-w^4)$ and $(t^3-w^6)$ and so the \lq virtual' case $(2,1)$ needs to be considered.
\subsubsection{Residues of the $w$-singularities}
We now evaluate the residues of $F(w,t)$ at the singularities $w$ that lie inside the unit circle, dealing with each case $(p,q)$ in turn.  We write $W_{p.q}(t)$ for the sum of the residues at those points $w$ where $w^p=t^q$ , excluding those already found from any earlier $(p',q')$ with $p/q=p'/q'$. 
\msk

\noindent\fbox{$(p,q)=(1,1)$}
\begin{equation}
W_{1,1}(t)=\frac{t^4U_{1,1}(t)}{36\,V_{1,1}(t)}
\end{equation}
where $U_{1,1}(t)$ is a palindromic polynomial in $t$ of degree $56$ with positive integer coefficients, and
\begin{align}
V_{1,1}(t)&=(1-t)(1-t^2)^3(1-t^3)^2(1-t^4)^3(1-t^5)^3(1-t^6)^3 
\notag \times\\
&\qquad \times(1-t^7)^2(1-t^8)(1-t^9).
\end{align}
\fbox{$(p,q)=(2,1)$}
\begin{equation}
W_{2,1}(t)=\frac{U_{2,1}(t)}{6\,V_{2,1}(t)}
\end{equation}
where $U_{2,1}(t)$ is a palindromic polynomial in $t$ of degree $60$ with positive integer coefficients, and
\begin{align}
V_{2,1}(t)&=(1-t)^2(1-t^2)^2(1-t^3)^4(1-t^4)^3(1-t^5)^3(1-t^6)
\notag \times\\
&\qquad \times(1-t^7)^2(1-t^9)(1-t^{11}).
\end{align}
\fbox{$(p,q)=(1,2)$}
\begin{equation}
W_{1,2}(t)=-\frac{2t^{19}U_{1,2}(t)}{V_{1,2}(t)}
\end{equation}
where $U_{1,2}(t)$ is a palindromic polynomial in $t$ of degree $46$ with positive integer coefficients, and
\begin{align}
V_{1,2}(t)& =(1-t)^2(1-t^2)(1-t^3)^2(1-t^4)^2(1-t^5)^3(1-t^6)^2(1-t^7)^2
\notag \times\\
&\qquad \times (1-t^8)(1-t^9)^2(1-t^{11})(1-t^{13}).
\end{align}
\fbox{$(p,q)=(2,2)$} \quad excluding residues already occurring at $(p,q)=(1,1)$:
\begin{equation}
W_{2,2}(t)=-\frac{t^4U_{2,2}(t)}{4\,V_{2,2}(t)}
\end{equation}
where $U_{2,2}(t)$ is a palindromic polynomial in $t$ of degree $20$ with positive integer coefficients, and
\begin{align}
V_{2,2}(t)& =(1-t)^2(1-t^2)(1-t^3)(1-t^4)^2(1-t^5)(1-t^6)(1-t^7)
\notag \times\\
&\qquad \times (1-t^8)(1-t^{12}).
\end{align}
\fbox{$(p,q)=(3,2)$}\quad  Here the calculations are simplified by noting that $(3,2)$ appears in Table~\ref{tab:pq} only where $k=-2,-1$. It also turns out that the $(3,2)$ residue sum for $R_{j,-2}$ is equal to that for $R_{j,-1}$, $j=1,2$.  We find
\begin{equation}
W_{3,2}=\frac{2t^2U_{3,2}}{V_{3,2}}
\end{equation}
where $U_{3,2}(t)$ is a palindromic polynomial in $t$ of degree $54$ with positive integer coefficients, and
\begin{align}
V_{3,2}(t)&=(1-t)^5(1-t^2)^3(1-t^3)(1-t^4)^2(1-t^5)^3(1-t^7)^2(1-t^8)
\notag \times\\
&\qquad \times (1-t^{11})(1-t^{13}).
\end{align}
\fbox{$(p,q)=(4,2)$}\quad  Here $(4,2)$ appears in Table~\ref{tab:pq} only where $k=-2$. Excluding residues already occurring at $(p,q)=(2,1)$ we find
\begin{equation}
W_{4,2}=\frac{U_{4,2}}{2\,V_{4,2}}
\end{equation}
where $U_{4,2}(t)$ is a palindromic polynomial in $t$ of degree $8$ with positive integer coefficients, and
\begin{equation}
V_{4,2}(t)=(1+t)^3(1-t^2)(1+t^3)^2(1-t^4)^2(1-t^6)(1-t^8).
\end{equation}
\fbox{$(p,q)=(1,3)$}
\begin{equation}
W_{1,3}=\frac{2t^{32}U_{1,3}}{V_{1,3}}
\end{equation}
where $U_{1,3}(t)$ is a palindromic polynomial in $t$ of degree $16$ with positive integer coefficients, and
\begin{align}
V_{1,3}(t)&=(1-t^2)^3(1-t^3)(1-t^4)^3(1-t^5)^3(1-t^6)^3(1-t^7)^2(1-t^8)
\notag \times\\
&\qquad \times (1-t^9)^2(1-t^{11}).
\end{align}
\fbox{$(p,q)=(2,3)$}
\begin{equation}
W_{2,3}=\frac{2t^{11}U_{2,3}}{V_{2,3}}
\end{equation}
where $U_{2,3}(t)$ is a palindromic polynomial in $t$ of degree $52$ with positive integer coefficients, and
\begin{align}
V_{2,3}(t)&=(1-t)^3(1-t^2)(1-t^3)^2(1-t^4)^2(1-t^5)^3(1-t^6)^2(1-t^7)^2
\notag \times\\
&\qquad \times (1-t^8)(1-t^9)^2(1-t^{13}).
\end{align}
\fbox{$(p,q)=(3,3)$}\quad excluding residues already occurring at $(p,q)=(1,1)$:
\begin{equation}
W_{3,3}=\frac{2t^4(1+t)(1+t^2)U_{3,3}}{9 V_{3,3}}
\end{equation}
where $U_{3,3}(t)$ is a palindromic polynomial in $t$ of degree $14$ with positive integer coefficients, and
\begin{equation}
V_{3,3}(t)=(1-t^3)(1-t^5)(1-t^6)^2(1-t^9)^2(1-t^{12}).
\end{equation}
\fbox{$(p,q)=(4,3)$}
\begin{equation}
W_{4,3}=-\frac{2t^2U_{4,3}}{V_{4,3}}
\end{equation}
where $U_{4,3}(t)$ is a palindromic polynomial in $t$ of degree $50$ with positive integer coefficients, and
\begin{align}
V_{4,3}(t)=(1-t)^5(1-t^2)^2(1-t^3)^4(1-t^5)^3(1-t^7)^2(1-t^9)^2(1-t^{11}).
\end{align}
\fbox{$(p,q)=(5,3)$}
\begin{equation}
W_{5,3}=-\frac{2t\,U_{5,3}}{V_{5,3}}
\end{equation}
where $U_{5,3}(t)$ is a palindromic polynomial in $t$ of degree $82$ with positive integer coefficients, and
\begin{align}
V_{5,3}(t)&=(1-t)^2(1-t^2)^4(1-t^3)(1-t^4)^2(1-t^6)^2(1-t^7)^2(1-t^8)(1-t^9)^2 \notag \times\\
&\qquad \times (1-t^{11})(1-t^{12})(1-t^{13}).
\end{align}
\fbox{$(p,q)=(6,3)$}\quad The factor $(w^6-t^3)$ appears in $R_{j,k}$ only for $k=-2$, and the residue sums (excluding those already obtained at $(p,q)=(2,1)$) for $R_{1,-2}$ and for $R_{2,-2}$ coincide and give 
\begin{equation}
W_{6,3}=\frac{4\,U_{6,3}}{3\,V_{6,3}}
\end{equation}
where $U_{6,3}(t)$ is a palindromic polynomial in $t$ of degree $12$ with positive integer coefficients, and
\begin{equation}
V_{6,3}(t)=(1-t)(1-t^3)^2(1-t^6)^2(1-t^9)^2.
\end{equation}
Finally, adding all these residues yields the following result.
\begin{prop}
The sum $W(t)$ of all the residues of $F(w,t)$ for $|w|<1$ (where $|t|<1$) has the form
\begin{equation}
W(t)=4\frac{U(t)}{V(t)}
\end{equation}
where $U(t)$ is a palindromic polynomial in $t$ of degree $96$ with integer coefficients
and
\begin{align}
V(t)&=(1+t)(1-t^2)^3(1-t^3)(1-t^4)^2(1-t^5)^3(1-t^6)^2(1-t^7)^2(1-t^8) \notag \times\\
&\qquad \times (1-t^9)^2(1-t^{11})(1-t^{12})(1-t^{13}).
\end{align}
\end{prop}
The algebraic form of $W(t)$ is somewhat anomalous for this context, as the coefficients in the numerator $U(t)$ are not all positive, and the denominator $V(t)$ has an apparently superfluous factor $(1+t)$.  Nevertheless, the total number of other factors (with repeats) is $19$ which is the codimension $25-6$ of the generic orbits of the group action as would be expected.  In fact expansion as a power series in $t$ yields a pleasing result.
\begin{theo}  \label{t:molienser0}
The Molien series $\PP^2_{{\mathbf\Gamma}_0}(t)$ for the action of ${\mathbf\Gamma}_0=\SO(3)\times\SO(3)$ on $L(V,V)\cong\R^{25}$ is given by
\begin{align}
\PP_{{\mathbf\Gamma}_0}^2(t)&=1+t^2+t^3+5t^4+5t^5+19t^6+27t^7+76t^8+136t^9+330t^{10}+626 t^{11}+\notag\\
&\quad+1391 t^{12}+2676 t^{13}+5497 t^{14}+10425 t^{15} + 20201 t^{16} + 37182 t^{17} \notag \\
&\qquad+68713t^{18} + 122489t^{19} + 217275t^{20} +  O(t^{21}).  \label{e:D20series}
\end{align}
\end{theo}
This form with positive integer coefficients suggests that suitable algebraic manipulation of $W(t)$ ought to put this rational function into a form consistent with data for primary and secondary generators of the ring of ${\mathbf\Gamma}_0$-invariants.  We are grateful to Boris Zhilinskii who showed us how to arrive at the following more informative conclusion.  

\begin{theo}  \label{t:molien1}
A rational Molien function $\PP^2_{{\mathbf\Gamma}_0}(t)$ for the action of ${\mathbf\Gamma}_0=\SO(3)\times\SO(3)$ on $L(V,V)\cong\R^{25}$ is given by
\begin{equation}
\PP_{{\mathbf\Gamma}_0}^2(t)=\frac{P_0(t)}{Q_0(t)}
\end{equation}
where
\begin{align}
 Q_0(t)&= (1-t^2)(1-t^3)(1-t^4)^3(1-t^5)(1-t^6)^2(1-t^7)^2(1-t^8)^2 \times  \notag \\
    &\qquad\times (1-t^9)^2(1-t^{10})^2(1-t^{11})(1-t^{12})(1-t^{13}) 
\end{align}
and $P_0(t)$ is a palindromic polynomial in $t$ of degree $113$ with positive integer coefficients as set out in \ref{app:molien_details}.
\end{theo}
For the Molien integral~(\ref{e:int+int}) for the full group ${\mathbf\Gamma}=\SO(3)\wr \Z_2$ in the $\ell=2$ case it follows from~(\ref{e:G1case2}) and~(\ref{e:stardenom}) that we must evaluate the residues at singularities inside the unit circle of the function
\begin{equation}
H(w,t):=\frac{(1-v)^2}{v^2\,K(v,t)}
\end{equation}
with $K(v,t)$ as in~(\ref{e:stardenom}), namely at the points $v=t^2,\pm t,\pm\sqrt{t}$ and $\exp(2m\pi i/3)t^{\frac23}$ for $m=0,1,2$, and finally divide by the factor $(1-t)^{3}(1+t)^{2}$.  We arrive at the following result. 
\begin{prop}
\begin{equation*}
\PP^2_{{\mathbf\Gamma}_1}(t)=\frac{U_1(t)}{V_1(t)}
\end{equation*}
\noindent where $U_1(t)$ is a palindromic polynomial in $t$ of degree $42$ with integer coefficients and
\begin{align}
V_1(t)&=(1+t^2)(1-t)(1-t^3)^2(1-t^4)^2(1-t^5)(1-t^6)^2(1-t^7)
\times \notag \\
&\qquad \times (1-t^8)^2(1-t^{10}).
\end{align}
\end{prop}
Again this may be put into a more propitious form
\begin{equation}  \label{e:p1q1}
\PP^2_{{\mathbf\Gamma}_1}(t)=\frac{P_1(t)}{Q_1(t)}
\end{equation}
where $Q_1(t)=Q_0(t)$ and where $P_1(t)$ is a polynomial in $t$ of degree $113$ which is palindromic up to sign-reflection as also given in \ref{app:molien_details}.
\msk

Expanding $\PP^2_{{\mathbf\Gamma}_1}(t)$ as a power series in $t$ we find
\begin{align}
\PP^2_{{\mathbf\Gamma}_1}(t)&=1+t^2+t^3+3t^4+5t^5+9t^6+13t^7+28t^8+44t^9+72t^{10}+116t^{11}+\notag\\
&\quad +193t^{12}+294t^{13}+457t^{14}+689t^{15}+1039t^{16}+1526t^{17}  \notag \\
&\qquad + 2221t^{18} + 3177t^{19} + 4541t^{20} +  O(t^{21})  \label{e:D21series}
\end{align}
which finally gives the Molien series for $\mathbf{\Gamma}$ as the arithmetic mean of \eqref{e:D20series} and \eqref{e:D21series} as follows:
\begin{theo}  \label{t:molienser}
\begin{align}
\PP^2_{\mathbf{\Gamma}}(t)&=1+t^2+t^3+ 4t^4+ 5t^5+ 14t^6+ 20t^7+ 52t^8+ 90t^9 +201t^{10}+ 371t^{11}+\notag\\
     &\quad + 792t^{12}+ 1485t^{13} +2977t^{14}+5557t^{15}+ 10 620t^{16}+ 19 354t^{17}\notag \\
&\qquad + 35 467t^{18} + 62833t^{19} + 110908t^{20} + O(t^{21}).  \label{e:mseries}
\end{align}
\end{theo}

Observe that (as must be the case since ${\mathbf\Gamma}_0<{\mathbf\Gamma}$) the positive integer coefficients in
$\PP^2_{\mathbf\Gamma}$ are less than or equal to the corresponding coefficients in $\PP^2_{{\mathbf\Gamma}_0}(t)$ as in Theorem~\ref{t:molienser0}: their difference at degree $n$ is the number of independent homogeneous invariants of degree $n$ for ${\mathbf\Gamma}_0$ that are {\it not} invariant under the transposition action. We examine this a little more closely in the following section.
\msk

Note also that when \eqref{e:mseries} is expressed as a rational function with denominator $Q_0(t)$ the numerator is $\tfrac12(P_0(t)+P_1(t))$.  The sum $\tfrac12(P_0(1)+P_1(1))$ of the coefficients of this polynomial is the total number of secondary invariants (linearly independent but algebraically dependent on the $19$ primary invariants), and turns out to be 726,963,024 as indicated in Section~\ref{ss:wreath}.

\section{The algebra of irreducible super-tensors}
\subsection{Tensor products and invariants}
The tensor product of two (complex) unitary representations $D,D'$ of a compact Lie group $\G$ on linear spaces $V,V'$ respectively is (as is any other) isomorphic to a direct sum of irreducible representations of $\G$:
\begin{equation}
D\otimes D' \cong D_0\oplus D_1\cdots\oplus D_m.
\end{equation}
Explicitly this means that with respect to a chosen basis of $V\tp V'$ there is a unitary $nn'\times nn'$ matrix $C$, independent of $g\in \G$, such that 
\begin{equation} \label{e:blocks}
C\,\big(D(g)\tp D'(g)\big)\,C^*=\diag\bigl(D_0(g),D_1(g),\ldots,D_m(g)\bigr)
\end{equation}
for all $g\in \G$, where the right hand side is a block-diagonal matrix.
\msk   

Suppose that $D_0$ is the trivial $1$-dimensional representation of $\G$, and that it appears precisely once in~(\ref{e:blocks}). Orthogonal projection of $V\tp V'$ onto its $1$-dimensional $\G$-invariant linear subspace corresponding to $D_0$ is a $\G$-invariant linear function on $V\tp V'$, which in the case $V=V'$ gives a quadratic $\G$-invariant function on $V$. 
More generally, for $k=0,\ldots, m$ the orthogonal projection of $V\tp V'$ onto its invariant $n_k$-dimensional subspace $V_k$ supporting the representation $D_k$ is an equivariant linear map given by the $n_k\times N$ matrix $C_k$ consisting of the $n_k$ rows of $C$ corresponding to the $D_k$-block. Using this projection and 
iterating the procedure with higher order tensor products using~(\ref{e:blocks}) provides a systematic method for finding higher order invariants. We illustrate the use of this method below.
\msk

For our purposes we need to extend these considerations to the situation where $V,V'$ are replaced by $L(W^*,V),L(W'^*,V')$ for further representation spaces $W,W'$ for $\G$. In view of the natural identification
\begin{equation}
L(W^*\tp W'^*,V\tp V')\cong L(W^*,V)\tp L(W'^*,V')
\end{equation}
for linear spaces $V,W,V',W'$ through which the element $A\tp B$ on the right hand side corresponds to the linear map
\begin{equation}
W^*\tp W'^* \to V\tp V' : w^*\tp w'^* \mapsto Aw^*\tp Bw'^*
\end{equation}
on the left hand side, the tensor product
representation $D\tp D'$ of $\G$ on $V\tp V'$ defines a natural action of $\G$ on $L(W^*,V)\tp L(W'^*,V')$ by left multiplication (right composition) on $L(W^*\tp W'^*,V\tp V')$, and similarly two representations $E,E'$ of $\G$ on $W,W'$ define a $\G$-action on $L(W^*,V)\tp L(W'^*,V')$ by right dual multiplication (left dual composition). Schematically, we have for $S\in L(W^*,V)$ and $S'\in L(W'^*,V')$:
\begin{center}
\includegraphics[height=3.5cm]{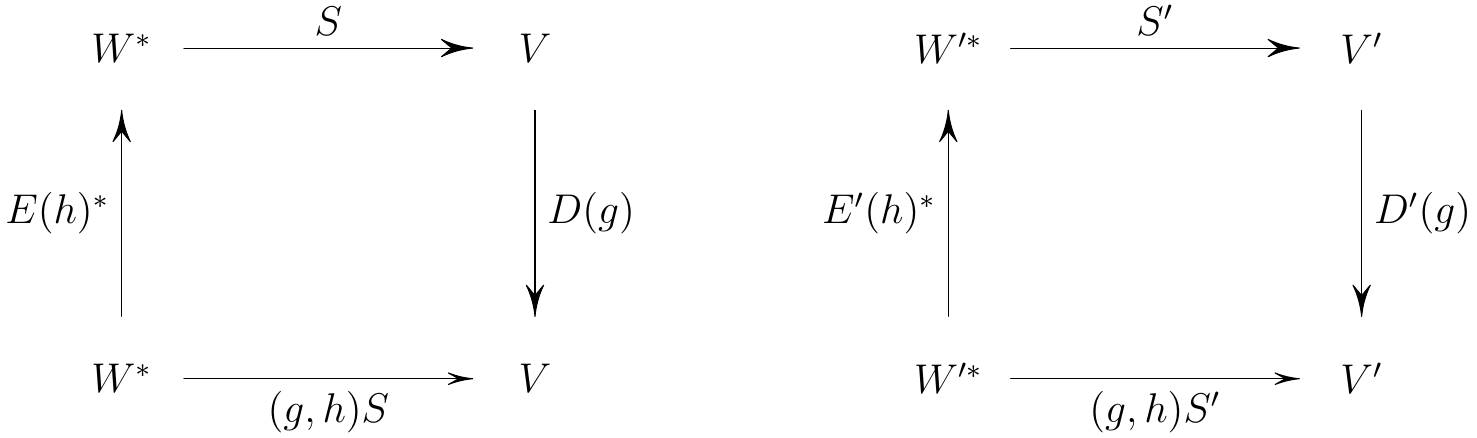} \, .
\end{center}
Together these define a representation of $\G\times \G$ on $L(W^*,V)\tp L(W'^*,V')$ given explicitly by
\begin{align}
(g,h):S\tp S' &\mapsto  \big[ D(g) S E^*(h)) \big] 
\tp \big[ D'(g) S' E'^*(h)) \big] \notag \\
& = \big[ D(g) \tp D'(g) \big] (S\tp S') 
\big[E^*(h) \tp E'^*(h) \big] \, . \label{e:bigrep}
\end{align}
Stated otherwise, the following diagram commutes:
\begin{center}
\includegraphics[height=3.5cm]{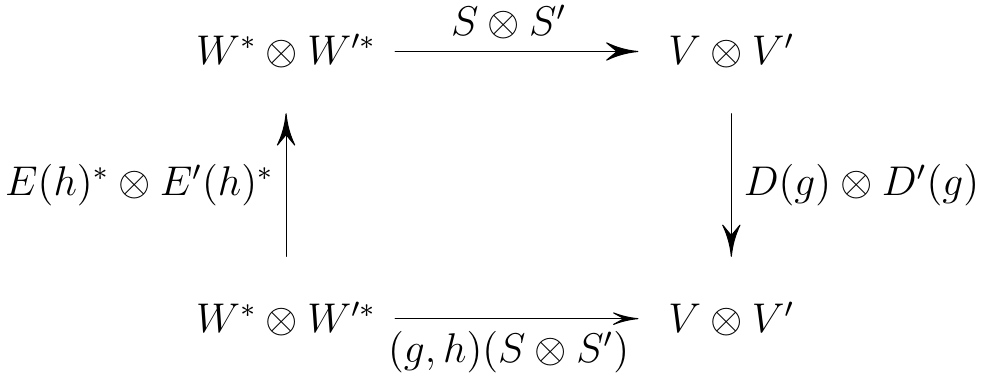}\, .
\end{center}
Let the projections $\{B_\ell:W\tp W'\to W_\ell\}$ be analogs for $E\tp E'$ of the projections $\{C_k:V\tp V'\to V_k\}$ for $D\tp D'$ described above. 
\begin{lemma}  \label{l:phimap}
For each pair $k,\ell$ the linear map
\begin{align}
\Phi_{k\ell}:L(W^*,V)\tp L(W'^*,V')\to L(W_\ell^*,V_k):S\tp S' \mapsto C_k\,S\tp S'\,B_\ell^{*} \notag
\end{align}
is $\G\times \G$-equivariant.
\end{lemma}
\begin{proof}
From the expression~(\ref{e:blocks}) we have
\begin{equation*} 
C_k\,\big(D(g)\tp D'(g)\big) = D_k(g)\,C_k
\end{equation*}
and likewise
\begin{equation*} 
\big(E^*(g)\tp E'^*(g)\big)\,B_\ell^* = B_\ell^*\,E_\ell^*(g)
\end{equation*}
for each pair $k,\ell$.  Therefore for $(g,h) \in \G\times \G$ we have
\begin{align*}
\Phi_{k\ell}((g,h)S\tp S')&= C_k\big[ D(g) \tp D'(g) \big] (S\tp S') \big[ E^*(h) \tp E'^*(h) \big] B_\ell^{*} \notag \\
  &=D_k(g)\,C_k \,(S\tp S')\,B_\ell^*\,E_\ell^*(h) \\
  &=(g,h)\Phi_{k\ell}(S\tp S').
\end{align*}
\end{proof}
We shall apply this in the case of $\G=\SO(3)$ with $V=V'=W=W'\cong\R^5$ and with $D,D',E,E'$ all given by the $j=2$ irreducible representation $D^{(2)}$.  The $j=0$ component of~(\ref{e:bigrep}) defined by $\Phi_{00}$ yields a quadratic invariant for the action of $\SO(3)\times\SO(3)$ on $L(V,V)$, and we show how further iteration of this construction yields explicit invariants of degree~3 and~4.
\msk

\subsection{The case of $\SO(3)$ and Clebsch-Gordan coefficients}
The context in which $\G=\SO(3)$ and $D,D'$ are irreducible is thoroughly studied in the Quantum Mechanics literature~\cite{Rose,Edmonds}.  If $D=D^{(j_1)}$ and $D'=D^{(j_2)}$ where $D^{(j_1)}, D^{(j_2)}$ are the unique irreducible representations on linear spaces $V^{(j_1)},V^{(j_2)}$ of dimension $2j_1+1,2j_2+1$ respectively, then the {\it Clebsch-Gordan formula} or {\it series}~\cite{BT,Hamermesh} 
states that 
\begin{equation}  \label{e:clebschf}
D^{(j_1)}\otimes D^{(j_2)}\cong D^{(|j_1 - j_2|)} \oplus D^{(|j_1 - j_2|+1)} \oplus \cdots \oplus D^{(j_1 + j_2)}
\end{equation}
or, more explicitly, that with spherical harmonics $Y_{jm},m=-j,\ldots,j$ chosen as a basis for $V^{(j)}$ and giving the basis 
$\{Y_{j_1 m_1} \tp Y_{j_2 m_2}\}$ for $V^{(j_1)}\otimes V^{(j_2)}$
there exists a unique unitary $N\times N$ matrix $\CC$, where $N=(2j_1+1)(2j_2+1)$, such that 
\begin{align}
\CC \big[\DD^{(j_1)} \tp \DD^{(j_2)}\big] \CC^{*} = 
\begin{pmatrix}
\DD^{(|j_1 - j_2|)} & 0 & \cdots & 0 \\
0 & \DD^{(|j_1 - j_2|+1)} & \cdots & 0 \\
\vdots & \vdots & \ddots & \vdots \\
0 & 0 & \cdots & \DD^{(j_1 + j_2)}
\end{pmatrix} \, .
\label{eq:CC_matrixform}
\end{align}
where the matrix $\DD^{(j)}$ is called the Wigner rotation matrix of rank $j$. The entries of the matrix $\CC$ are known as {\it Clebsch-Gordan coefficients}~\cite{Rose,Hamermesh} and are usually written as
\begin{equation}
\CC_{jm;m_1 m_2} = \cg{jm}{j_1 m_1; j_2 m_2} \, , 
\end{equation}
where for fixed $j_1,j_2$ the pairs of indices $j,m$ and $m_1, m_2$ label rows and columns respectively. These coefficients vanish unless
\begin{subequations}
\begin{align}
\qquad |j_1 - j_2| \leq j \leq j_1 + j_2 
\, , \qquad m_1 + m_2 = m \qquad \quad \\
-j_1 \leq m_1 \leq j_1 
\, , \qquad -j_2 \leq m_2 \leq j_2 
\, , \qquad -j \leq m \leq j \, .
\end{align}
\label{e:cgprops}
\end{subequations}
Note that the right hand side of~(\ref{e:clebschf}) includes the trivial $1$-dimensional representation $D^{(0)}$ if and only if $j_1=j_2$.
\msk

The use of the explicit matrix representation in terms of spherical harmonics is particularly convenient since many detailed results are readily available in the Quantum Mechanics literature.
For later reference it is also useful to introduce the more symmetric  Wigner $3$-$j$ symbols~\cite{Edmonds} defined as
\begin{align}
\threej{j_1 & j_2 & j_3}{m_1 & m_2 & m_3} & 
= \frac{(-1)^{j_1-j_2-m_3}}{\sqrt{2j_3+1}} 
\cg{j_3 (-m_3)}{j_1 m_1;j_2 m_2} \, .
\end{align}
These vanish unless
\begin{equation}m_1 + m_2 + m_3 = 0\, , \qquad 
|j_1 -j_2| \leq j_3 \leq j_1 + j_2 \, .
\end{equation}
Finally, we record the following special case of the Clebsh-Gordan coefficients
\begin{align}
\threej{j_1 & j_2 & 0}{m_1 & m_2 & 0} = 
\cg{0 0}{j_1 m_1 ; j_2 m_2} = \delta_{j_1,j_2}\delta_{m_1,-m_2} \frac{(-1)^{j_1-m_1}}{\sqrt{2j_1+1}} \, ,
\label{eq:cg_special_case}
\end{align}
where $\delta_{h,k}$ is the Kronecker symbol equal to $1$ if $h=k$ and zero otherwise. We refer to any standard textbook on Quantum Mechanics for a more complete review of the numerous interesting identities and symmetries of Clebsch-Gordan coefficients and $3$-$j$ symbols.
\subsection{Construction of invariants for the $\SO(3)\times\SO(3)$-action on $L(V^{(2)},V^{(2)})$}   
We now focus on the case $j_1=j_2=2$. Let $\CC_{j}$ denote the rectangular matrix obtained from $\CC$ by selecting the $(2j+1)$-rows associated with the $j^{\text{th}}$ block in~\eqref{eq:CC_matrixform}, so that $\CC_{j}$ has dimensions $(2j+1) \times 25$.  Taking $j=0$ (so $m=0$) and using~(\ref{e:cgprops}) we see from Lemma~\ref{l:phimap} with $k,\ell=0,0$ and Eq.\eqref{eq:cg_special_case} that for $S\in L(V^{(2)},V^{(2)})$
\begin{align}
\Phi_{00}(S\tp S)& =\sum_{m,m}(-1)^{m+m'} \cg{00}{2 m_1; 2 m_2} \cg{00}{2 m_1'; 2 m_2'} S_{m_1 m_1'} S_{m_2\, m_2'} \notag \\
& = \frac{1}{5} \sum_{m,m}(-1)^{m+m'} S_{m m'} S_{-m\, -m'} \, ,
\end{align}
with $-2\le m,m'\le 2$. Furthermore, since 
any Wigner rotation matrix $A=\DD^{(\ell)}$ automatically satisfies
\[
A_{m m'}^*=(-1)^{m+m'}A_{-m\,-m'}
\]
for $-\ell\le m,m'\le\ell$, the order parameter matrix $S$ (obtained as an ensemble average over such~$A$: see~\cite{2011ST}) has the same properties
and thus we immediately obtain 
\begin{equation}
\Phi_{00}(S\tp S) = I_2:= \|S\|^2 \, , \label{eq:invariant_2}
\end{equation}
as the expected quadratic invariant.
\msk

With a little more effort we can construct a cubic invariant. To this end, we apply the same process as above at degree~3 to obtain:
\begin{align}
\Phi_{00} [\Phi_{jj'}(S\tp S) \tp S]  = \sum_{m,m',m_3,m_3'}& \cg{00}{j m; 2 m_3} \cg{00}{j' m'; 2 m_3'} \times \notag \\
& \qquad \times[\Phi_{jj'}(S\tp S)]_{mm'} S_{m_3 m_3'} \, .
\end{align}
After the substitution of Eq.\eqref{eq:cg_special_case}, we can see that the only non vanishing cubic invariant comes from the choice $j=j'=2$. Thus, we obtain the following cubic invariant (see also \cite{Sattinger}):
\begin{align}
I_3 = \sum_{m_i, m_i'=-2}^{2} \threej{2 & 2 & 2}{m_1 & m_2 & m_3} \threej{2 & 2 & 2}{m_1' & m_2' & m_3'}
 S_{m_1 m_1'} S_{m_2 m_2'} S_{m_3 m_3'} \, . \label{eq:invariant_3}
\end{align}  
From the Molien series $\PP^2_{{\mathbf\Gamma}_0}$ in Theorem~\ref{t:molienser0} we expect a single generator for ${\mathbf\Gamma}_0$-invariants at each homogeneous degree~2 and~3, and so $I_3$ is the unique degree-$3$ invariant up to scalar multiplication. Moreover, we observe that $I_2$ and $I_3$ are clearly also invariant with respect to the transposition $\tau$ which exchanges $m_i$ and $m_i'$. Therefore we conclude the following result.
\begin{prop}
The invariants $I_2,I_3$ generate the quadratic and the cubic invariants of the whole group $\mathbf{\Gamma} \cong \SO(3) \wr\Z_2$.
\end{prop}

The structure of the invariants is, however, more involved at degree four. To describe it more closely, let us first define the following super-tensors ($j,j' = 0,\ldots,4$)
\begin{align}  \label{e:Udef}
U^{(j,j')}_{k k'} = \sum_{\substack{|m_i|\leq j \\ |m_i'|\leq j'}} 
\threej{2 & 2 & j}{m_1 & m_2 & k} \threej{2 & 2 & j'}{m_1' & m_2' & k'}
S_{m_1 m_1'} S_{m_2 m_2'} \, .
\end{align}
The $\SO(3) \times \SO(3)$ invariants of degree~$4$ can then be written as contractions (scalar product) of the $U^{(j,j')}$:
\begin{equation}  \label{e:I4def}
I_{4}^{(j,j')} = \sum_{\substack{|k|\leq j \\ |k'|\leq j'}} (-1)^{k+k'} U^{(j,j')}_{k k'} U^{(j,j')}_{-k,-k'} \, ,
\end{equation}
of which there are 25 since both $j$ and $j'$ range from $0$ to $4$. However, as we next show, these are not independent. Let us introduce the following definitions for the symmetric and skew-symmetric part of $I_{4}^{(j,j')}$:
\begin{align}
I_{4}^{[j,j']} = \frac{1}{2} \big(I_{4}^{(j,j')} + I_{4}^{(j',j)} \big) \, , \quad 
I_{4}^{\{j,j'\}} = \frac{1}{2} \big(I_{4}^{(j,j')} - I_{4}^{(j',j)} \big) \, .
\end{align}
As described in detail in \ref{app:quartic}, the following general relations hold:
\begin{subequations}
\begin{align}
I_{4}^{(j,j')} & = 0 \quad \text{ if } j+j' \text{ is odd, }  \label{eq:I4_rel_1} \\
\big( I_{4}^{(j,j')} \big)^{*} & = I_{4}^{(j,j')} \, , \label{eq:I4_rel_2}\\
\tau I_{4}^{[j,j']} & = I_{4}^{[j,j']} \, , \qquad \tau I_{4}^{\{j,j'\}} = -I_{4}^{\{j,j'\}} \label{eq:I4_rel_3} \, .
\end{align}
\end{subequations}
There are 12 relations given by~(\ref{eq:I4_rel_1}) which reduces the number of potentially independent invariants~(\ref{e:I4def}) from~25 to~13.  Moreover, as can be checked directly, the following less natural identities also hold:
%
%
\begin{subequations}
\begin{align}
4 \, I_{4}^{(0,0)} + 9 \, I_{4}^{(1,1)} + 5\, I_{4}^{(2,2)} 
- 14\, I_{4}^{(3,3)} - 54\, I_{4}^{(4,4)} & = 0 \, , \\
60\, I_{4}^{(0,0)} + 9\, I_{4}^{(1,1)} + 245\, I_{4}^{(2,2)}
-784\, I_{4}^{(3,3)} - 280\, I_{4}^{[0,2]} & = 0 \, , \\
212\, I_{4}^{(0,0)} - 909\, I_{4}^{(1,1)} + 2695\, I_{4}^{(2,2)}
-3136\, I_{4}^{(3,3)} - 1512\, I_{4}^{[0,4]} & = 0 \, , \\
100\, I_{4}^{(0,0)} + 99\, I_{4}^{(1,1)} - 1225\, I_{4}^{(2,2)}
+784\, I_{4}^{(3,3)} - 1008\, I_{4}^{[1,3]} & = 0 \, , \\
220\, I_{4}^{(0,0)} - 387\, I_{4}^{(1,1)} - 535\, I_{4}^{(2,2)}
+ 112\, I_{4}^{(3,3)} - 2160\, I_{4}^{[2,4]} & = 0 \, , \\
5\, I_{4}^{\{0,2 \}} - 9\, I_{4}^{\{0,4 \}} & = 0\, , \\
5\, I_{4}^{\{0,2 \}} - 6\, I_{4}^{\{1,3 \}} & = 0\, , \\
7\, I_{4}^{\{0,2 \}} + 18\, I_{4}^{\{2,4 \}} & = 0 \, .
\end{align}
\end{subequations}
See \ref{app:quartic} or further details.
These eight identities reduce the number of potentially independent degree-4 invariants from~13 to~5.  From the Molien series $\PP_{{\mathbf\Gamma}_0}^2(t)$ in Theorem~\ref{t:molienser0} we expect to have five linearly independent degree-4 invariants of ${\mathbf\Gamma}_0=\SO(3)\times \SO(3)$, and we may accordingly take these to be
$I_{4}^{(0,0)}$, $I_{4}^{(1,1)}$, $I_{4}^{(2,2)}$, $I_{4}^{(3,3)}$, $I_{4}^{\{0,2 \}}$, say.

The relations~\eqref{eq:I4_rel_3} show that, of these five invariants of ${\mathbf\Gamma}_0$, all except the last are invariant under the transposition operator~$\tau$.  This is consistent with the Molien series $\PP_{\mathbf\Gamma}^2(t)$ in Theorem~\ref{t:molienser} which predicts only four linearly independent degree-4 invariants of $\mathbf{\Gamma}=\SO(3)\wr \Z_2$. To summarise, we have established the following result:
\begin{prop}
The invariants $I_{4}^{(0,0)}$, $I_{4}^{(1,1)}$, $I_{4}^{(2,2)}$, $I_{4}^{(3,3)}$ generate the quartic invariants of the whole group $\mathbf{\Gamma} \cong \SO(3) \wr\Z_2$.
\end{prop}
\section*{Acknowledgements} The authors are indebted to Tim Sluckin for instructive conversations on liquid crystals.  
DC is grateful to the Leverhulme Trust for research support via an Emeritus Fellowship, and to BCAM, Bilbao for support and excellent working conditions.
ST thanks the Mathematical Sciences Department, University of Southampton, for hospitality, while DC and RL thank the Southampton/Hamburg Exchange Scheme for supporting several research visits in both directions.  The authors are also grateful to the Isaac Newton Institute, Cambridge, where part of this work was carried out during the Programme on the Mathematics of Liquid Crystals in 2013. 

\section*{References}
\bibliographystyle{unsrt}
\bibliography{references}

\appendix
\section{Molien function details}
\label{app:molien_details}
The numerator $P_0(t)$ of the rational Molien function in Theorem~\ref{t:molien1} is:
\begin{align}
P_0(t)&=
1 + t^{4}
+ 3 t^{5}
+ 11 t^{6}
+ 16 t^{7}
+ 42 t^{8}
+ 80 t^{9}
+ 185 t^{10}
+ 357 t^{11}
+ 752 t^{12}\notag \\
&+ 1412 t^{13}
+ 2723 t^{14}
+ 4937 t^{15}
+ 8888 t^{16}
+ 15 342 t^{17}
+ 26 146 t^{18}
+ 43 083 t^{19}\notag \\
&+ 69 884 t^{20}
+ 110 398 t^{21}
+ 171 406 t^{22}
+ 260 288 t^{23}
+ 388 723 t^{24}
+ 569 210 t^{25}\notag \\
&+ 820 356 t^{26}
+ 1 161 726 t^{27}
+ 1 620 330 t^{28}
+ 2 224 150 t^{29}
+ 3 009 500 t^{30}\notag \\
&+ 4 012 238 t^{31}
+ 5 276 926 t^{32}
+ 6 845 013 t^{33}
+ 8 764 870 t^{34}
+ 11 078 260 t^{35}\notag \\
&+ 13 830 477 t^{36}
+ 17 054 459 t^{37}
+ 20 782 913 t^{38}
+ 25 029 615 t^{39}\notag \\
&+ 29 802 829 t^{40}
+ 35 086 893 t^{41}
+ 40 855 850 t^{42}
+ 47 055 721 t^{43}
+ 53 620 919 t^{44}\notag \\
&+ 60 456 820 t^{45}
+ 67 458 001 t^{46}
+ 74 494 882 t^{47}
+ 81 431 353 t^{48}
+ 88 115 150 t^{49}\notag \\
&+ 94 396 925 t^{50}
+ 100 121 953 t^{51}
+ 105 148 447 t^{52}
+ 109 343 460 t^{53}
+ 112 595 858 t^{54}\notag \\
&+ 114 815 204 t^{55}
+ 115 941 062 t^{56}
+ 115 941 062 t^{57}
+ 114 815 204 t^{58}
+ 112 595 858 t^{59}\notag \\
&+ 109 343 460 t^{60}
+ 105 148 447 t^{61}
+ 100 121 953 t^{62}
+ 94 396 925 t^{63}
+ 88 115 150 t^{64}\notag \\
&+ 81 431 353 t^{65}
+ 74 494 882 t^{66}
+ 67 458 001 t^{67}
+ 60 456 820 t^{68}
+ 53 620 919 t^{69}\notag \\
&+ 47 055 721 t^{70}
+ 40 855 850 t^{71}
+ 35 086 893 t^{72}
+ 29 802 829 t^{73}
+ 25 029 615 t^{74}\notag \\
&+ 20 782 913 t^{75}
+ 17 054 459 t^{76}
+ 13 830 477 t^{77}
+ 11 078 260 t^{78}
+ 8 764 870 t^{79}\notag \\
&+ 6 845 013 t^{80}
+ 5 276 926 t^{81}
+ 4 012 238 t^{82}
+ 3 009 500 t^{83}
+ 2 224 150 t^{84}\notag \\
&+ 1 620 330 t^{85}
+ 1 161 726 t^{86}
+ 820 356 t^{87}
+ 569 210 t^{88}
+ 388 723 t^{89}
+ 260 288 t^{90}\notag \\
&+ 171 406 t^{91}
+ 110 398 t^{92}
+ 69 884 t^{93}
+ 43 083 t^{94}
+ 26 146 t^{95}
+ 15 342 t^{96}
+ 8888 t^{97}\notag \\
&+ 4937 t^{98}
+ 2723 t^{99}
+ 1412 t^{100}
+ 752 t^{101}
+ 357 t^{102}
+ 185 t^{103}
+ 80 t^{104}
+ 42 t^{105}\notag \\
&+ 16 t^{106}
+ 11 t^{107}
+ 3 t^{108}
+ t^{109}
+ t^{113}.
\end{align}
The numerator $P_1(t)$ of the rational Molien function in~(\ref{e:p1q1}) is:
\begin{align}
P_1(t) & =   
1 - t^4 + 3 t^5 + 3 t^6 + 4 t^7 + 10 t^8 + 12 t^9 + 17 t^{10} + 25 t^{11}+ 30 t^{12}+ 36 t^{13} \notag \\
& + 41 t^{14}+ 41 t^{15} + 40 t^{16}+ 36 t^{17}+ 16 t^{18}- 9 t^{19}- 32 t^{20} - 74 t^{21}- 122 t^{22} \notag \\
& - 168 t^{23}- 223 t^{24}- 266 t^{25} - 298 t^{26}- 324 t^{27}- 312 t^{28}- 274 t^{29}- 216 t^{30} \notag\\
&- 108 t^{31}+ 30 t^{32}+ 183 t^{33}+ 364 t^{34} + 546 t^{35}+ 717 t^{36}+ 871 t^{37}+ 961 t^{38} \notag \\
& + 999 t^{39}+ 979 t^{40} + 859 t^{41}+ 670 t^{42}+ 413 t^{43} + 83 t^{44}- 268 t^{45}- 639 t^{46} \notag \\
& - 1002 t^{47}- 1299 t^{48}- 1536 t^{49}- 1683 t^{50} - 1695 t^{51}- 1601 t^{52}- 1398 t^{53} \notag \\ 
& - 1072 t^{54} - 680 t^{55}- 238 t^{56}+ 238 t^{57}+ 680 t^{58}+ 1072 t^{59}+ 1398 t^{60} + 1601 t^{61} \notag\\
& + 1695 t^{62}+ 1683 t^{63}+ 1536 t^{64}+ 1299 t^{65}+ 1002 t^{66}+ 639 t^{67}+ 268 t^{68} - 83 t^{69} \notag\\
& - 413 t^{70} -670 t^{71}- 859 t^{72}- 979 t^{73}- 999 t^{74}- 961 t^{75}- 871 t^{76}- 717 t^{77} \notag\\
& - 546 t^{78}- 364 t^{79}-183 t^{80} - 30 t^{81}+ 108 t^{82}+ 216 t^{83}+ 274 t^{84}+ 312 t^{85} \notag\\
& + 324 t^{86}+ 298 t^{87}+ 266 t^{88}+ 223 t^{89}+ 168 t^{90} + 122 t^{91}+ 74 t^{92}+ 32 t^{93} \notag\\
& + 9 t^{94}- 16 t^{95}- 36 t^{96}- 40 t^{97}- 41 t^{98}- 41 t^{99}- 36 t^{100} - 30 t^{101}- 25 t^{102} \notag\\
& - 17 t^{103}- 12 t^{104}- 10 t^{105}- 4 t^{106}- 3 t^{107}- 3 t^{108}+ t^{109}- t^{113}.
\end{align}
\section{Quartic invariants}
\label{app:quartic}
To understand the various relations among the invariant quantities, it is helpful to study the symmetries of the $U^{(j,j')}_{k k'}$ as defined in~(\ref{e:Udef}) which arise as a consequence of the symmetries of the $3j$-symbols. In particular, we make use of the following identities~\cite{Edmonds}
\begin{align}
\threej{j_1 & j_2 & j_3}{m_1 & m_2 & m_3} & = \threej{j_2 & j_3 & j_1}{m_2 & m_3 & m_1} 
= \threej{j_3 & j_1 & j_2}{m_3 & m_1 & m_2} \label{eq:3jsymm_1}\\
\threej{j_1 & j_2 & j_3}{m_1 & m_2 & m_3} & 
= (-1)^{j_1+j_2+j_3} \threej{j_2 & j_1 & j_3}{m_2 & m_1 & m_3} \notag \\
& = (-1)^{j_1+j_2+j_3} \threej{j_1 & j_3 & j_2}{m_1 & m_3 & m_2} \label{eq:3jsymm_2}\\
\threej{j_1 & j_2 & j_3}{-m_1 & -m_2 & -m_3} & 
= (-1)^{j_1+j_2+j_3} \threej{j_1 & j_2 & j_3}{m_1 & m_2 & m_3} \, .
\label{eq:3jsymm_3}
\end{align}
We now show that the relations among the invariants as given in Eqs.\eqref{eq:I4_rel_1}--\eqref{eq:I4_rel_3} are simple consequences of the symmetries \eqref{eq:3jsymm_1}--\eqref{eq:3jsymm_3}:
\begin{enumerate}
\item 
\begin{align}
U^{(j,j')}_{k k'} & = 
(-1)^{j+j'}\sum_{\substack{|m_i|\leq j \\ |m_i'|\leq j'}} 
\threej{2 & 2 & j}{m_2 & m_1 & k} \threej{2 & 2 & j'}{m_2' & m_1' & k'} 
S_{m_2 m_2'} S_{m_1 m_1'} \notag \\
& = (-1)^{j+j'} U^{(j,j')}_{k k'} .
\end{align}
Therefore, $U^{(j,j')}_{k k'} = 0$ if $j+j'$ is odd, and so is the associated invariant.
\item 
\begin{align}
U^{(j,j')}_{k k'} & 
= (-1)^{j+j'}\sum_{\substack{|m_i|\leq j \\ |m_i'|\leq j'}} 
\threej{2 & 2 & j}{-m_1 & -m_2 & -k} \threej{2 & 2 & j'}{-m_1' & -m_2' & -k'}
S_{m_1 m_1'} S_{m_2 m_2'} \notag \\
& = (-1)^{j+j'}\sum_{\substack{|m_i|\leq j \\ |m_i'|\leq j'}} 
\threej{2 & 2 & j}{-m_1 & -m_2 & -k} \threej{2 & 2 & j'}{-m_1' & -m_2' & -k'} \times \notag \\
& \qquad \times (-1)^{m_1+m_1'+m_2+m_2'} S^*_{-m_1 -m_1'} S^*_{-m_2 -m_2'} \notag \\
& = (-1)^{j+j'+k+k'}\sum_{\substack{|m_i|\leq j \\ |m_i'|\leq j'}} 
\threej{2 & 2 & j}{-m_1 & -m_2 & -k} \threej{2 & 2 & j'}{-m_1' & -m_2' & -k'} \times \notag \\
& \qquad \times S^*_{-m_1 -m_1'} S^*_{-m_2 -m_2'} \notag \\
& = (-1)^{j+j'+k+k'} \big( U^{(j,j')}_{-k, -k'} \big)^{*} \, .
\end{align}
Since we know  from (i) that $j+j'$ must be even for non-vanishing $Us$, we are left with
\begin{equation}\big( U^{(j,j')}_{k k'}\big)^{*} = (-1)^{k+k'} U^{(j,j')}_{-k, -k'} \, .\end{equation}  
This has the immediate consequence from~\eqref{e:I4def} that the invariants $I_{4}^{(j,j')}$ are all real.
\item Finally, let us study how the $U^{(j,j')}_{k k'}$ transform under transposition,
\begin{align}
\hspace{-1cm}\big( \tau U^{(j,j')}\big)_{k k'} & = \sum_{\substack{|m_i|\leq j \\ |m_i'|\leq j'}} 
\threej{2 & 2 & j}{m_1 & m_2 & k} \threej{2 & 2 & j'}{m_1' & m_2' & k'}
\times \notag \\
& \qquad \times (-1)^{m_1+m_1'+m_2+m_2'} S_{-m_1', -m_1} S_{-m_2',-m_2} \notag \\
& = (-1)^{k+k'} \sum_{\substack{|m_i|\leq j \\ |m_i'|\leq j'}} 
\threej{2 & 2 & j'}{-m_1' & -m_2' & k'} \threej{2 & 2 & j}{-m_1 & -m_2 & k} S_{m_1' m_1} S_{m_2' m_2} \notag \\ 
& = (-1)^{k+k'+j+j'} \sum_{\substack{|m_i|\leq j \\ |m_i'|\leq j'}} 
\threej{2 & 2 & j'}{m_1' & m_2' & -k'} \threej{2 & 2 & j}{m_1 & m_2 & -k}
 S_{m_1' m_1} S_{m_2' m_2} \notag \\ 
& = (-1)^{k+k'+j+j'} U^{(j',j)}_{-k',-k} \, .
\end{align}
From this we obtain
\begin{align}
\tau I_{4}^{(j,j')} = \sum_{\substack{|k|\leq j \\ |k'|\leq j'}} (-1)^{k+k'} 
U^{(j',j)}_{-k',-k} U^{(j',j)}_{k',k}
= I_{4}^{(j',j)} \, .
\label{eq:tau_invariants}
\end{align}
Therefore, the ``diagonal'' invariants ($j=j'$) are left unchanged by the transposition. The ``non-diagonal'' invariants ($j \neq j'$), in principle, are not. In terms of the symmetric and anti-symmetric parts of $I_{4}^{(j,j')}$, Eq.\eqref{eq:tau_invariants} yields
\begin{equation}\tau I_{4}^{[j,j']} = I_{4}^{[j,j']} \, , \qquad 
\tau I_{4}^{\{j,j'\}} = -I_{4}^{\{j,j'\}} \, .\end{equation}
\end{enumerate}
\end{document}